\RequirePackage[l2tabu, orthodox]{nag} % Do not remove

\documentclass[letterpaper, 10pt, conference]{IEEEtran}

\usepackage[utf8]{inputenc}
\usepackage[T1]{fontenc}
\usepackage{pifont}
\usepackage{microtype}
\usepackage{mathtools}
\usepackage{amssymb}
\usepackage{amsthm}
\usepackage{amsmath}
\usepackage[inline]{enumitem}
\usepackage{xspace}
\usepackage[algoruled,linesnumbered]{algorithm2e}
\usepackage{multirow}
\usepackage{booktabs}
\usepackage{wrapfig}
\usepackage{subcaption}
\usepackage{tikz}
\usepackage[hidelinks]{hyperref}

\usetikzlibrary{automata, arrows, patterns}
\tikzset{auto}
\tikzset{shorten >=1pt, >=stealth}

\newcommand{\nat}{\mathbb N}

\newcommand{\playera}{Player~$0$}
\newcommand{\playerb}{Player~$1$}

\newcommand{\Pos}{\mathit{Pos}}
\newcommand{\Neg}{\mathit{Neg}}
\newcommand{\Ex}{\mathit{Ex}}
\newcommand{\Un}{\mathit{Un}}

\newcommand{\exampleend}{\hfill\tikz[baseline] \draw (0, 0) -| ++(1ex, 1.25ex);}

\newcommand{\consynth}{\textit{CONSYNTH}\xspace}
\newcommand{\dtsynth}{\textit{DT-Synth}\xspace}
\newcommand{\satsynth}{\textit{SAT-Synth}\xspace}
\newcommand{\rpnisynth}{\textit{RPNI-Synth}\xspace}

\newtheorem{definition}{Definition}
\newtheorem{example}{Example}
\newtheorem{theorem}{Theorem}
\newtheorem{lemma}{Lemma}
\IEEEoverridecommandlockouts

\begin{document}

%---------- Title ----------
\title{Learning-Based Synthesis of Safety Controllers}

%---------- Authors ----------
\author{\IEEEauthorblockN{Daniel Neider}
\IEEEauthorblockA{Max Planck Institute for Software Systems \\
Kaiserslautern, Germany \\
Email: neider@mpi-sws.org}
\and
\IEEEauthorblockN{Oliver Markgraf}
\IEEEauthorblockA{Max Planck Institute for Software Systems \\
Kaiserslautern, Germany}}

%---------- Title ----------
\maketitle

%---------- Abstract ----------
\begin{abstract}
We propose a machine learning framework to synthesize reactive controllers for systems whose interactions with their adversarial environment are modeled by infinite-duration, two-player games over (potentially) infinite graphs.
Our framework targets safety games with infinitely many vertices, but it is also applicable to safety games over finite graphs whose size is too prohibitive for conventional synthesis techniques.
The learning takes place in a feedback loop between a teacher component, which can reason symbolically about the safety game, and a learning algorithm, which successively learns an approximation of the winning region from various kinds of examples provided by the teacher.
We develop a novel decision tree learning algorithm for this setting and show that our algorithm is guaranteed to converge to a reactive safety controller if a suitable approximation of the winning region can be expressed as a decision tree.
Finally, we empirically compare the performance of a prototype implementation to existing approaches, which are based on constraint solving and automata learning, respectively.

\end{abstract}

%---------- Introduction ----------
% !Tex root=main.tex

\section{Introduction}
\label{sec:intro}

Reactive synthesis offers an effective and promising way to solve a crucial practical problem: constructing correct and verified controllers for safety-critical systems.
Rather than designing and implementing controllers by hand, reactive synthesis techniques construct controllers in an automatic fashion, thus, freeing engineers from this complex and error-prone task.
In addition to being fully automatic, synthesis techniques produce correct-by-construction controllers that guarantee to satisfy the given specification, or they report that no such controller exist.

Typically, reactive synthesis is modeled as an infinite-duration game on a graph that is played by two antagonistic players: the system, which seeks to satisfy the specification, and the environment, which wants to violate it.
More precisely, the specification and a model of the environment are in a first step converted into an infinite game.
Then, one computes a winning strategy for the system, which prescribes how the system needs to play in order to win against every move of the environment.
Finally, the winning strategy is translated into hard- or software, resulting in a reactive controller that satisfies the given specification.

In this paper, we focus on safety games, a class of infinite games that arises from safety specifications.
Such specifications are in fact among the most important in practice (e.g., see Dwyer, Avrunin, and Corbett~\cite{DBLP:conf/icse/DwyerAC99} for a survey of specification patterns) and capture many other interesting properties, including bounded-horizon reachability.
In contrast to the classical setting, however, we consider safety games not only over finite graphs but also over graphs with infinitely many (even uncountably many) vertices.
Such games arise naturally, for instance, when the interaction between the controlled system and its environment is too complex to be modeled by finite graphs (e.g., in motion planning over unbounded environments) or when the environment has access to dynamic data structures, such as lists, stacks, or queues.

When the number of vertices of a game graph is infinite, traditional methods, which typically rely on an explicit exploration of the (whole) graph, are no longer applicable.
To enable the computation of winning strategies for such games, our first contribution is a machine learning framework that constructs winning strategies by learning a proxy object called \emph{winning set}.
Intuitively, a winning set is an approximation of the vertices from which the system player can win and that permits to extract a winning strategy in a simple manner.

Our learning framework defines a feedback loop, akin to counterexample-guided inductive synthesis (CEGIS)~\cite{DBLP:conf/aplas/Solar-Lezama09}, consisting of two entities: a teacher, who can reason symbolically about the game, and a learning algorithm, whose goal is to learn a winning set from information provided by the teacher.
In every iteration of the loop, the learning algorithm constructs a set of vertices, which it proposes to the teacher.
The teacher, on the other hand, checks whether the proposed set is a winning set and stops the learning process if so.
If this is not the case, the teacher returns a counterexample.
Upon receiving a counterexample, the learner refines its conjecture and proceeds with the next iteration.

Motivated by recent success in using decision trees as a concise and effective representation of strategies in infinite games~\cite{DBLP:conf/tacas/BrazdilCKT18}, our second contribution is a new learning algorithm for decision trees, which is tailored specifically to the framework sketched above.
Our algorithm builds upon a recent learning algorithm for decision trees that has been proposed in the context of software verification~\cite{DBLP:journals/pacmpl/EzudheenND0M18} and that puts a strong emphasis on learning ``small'' trees.
As a consequence of the latter, our algorithm can in many situations guarantee to learn a winning set if one can be expressed as a decision tree.

Even though a game graph is infinite or prohibitively large, a reactive controller with a compact representation might already realize the specification.
In a motion planning scenario, for instance, the system often only needs to consider a small subset of possible interactions with the environment to satisfy the specification.
Based on this observation, our learning-based approach possesses various desirable properties:
\begin{enumerate*}[label={(\roman*)}]
	\item it leverages machine learning as an effective means to focus on the important parts of a game,
	\item it often learns ``small'' winning sets (and, by extension, small winning strategies) in a rule-like format, which tend to be relatively easy for humans to understand (see Brázdil et al.~\cite{DBLP:conf/tacas/BrazdilCKT18}), and
	\item besides operating over infinite graphs, it guarantees in many situations to learn a winning set if one exists.
\end{enumerate*}
In addition, we demonstrate empirically that our approach is highly competitive to existing tools on two sets of benchmarks taken from the literature.

% !Tex root=main.tex

\paragraph*{\bfseries Related Work}
Games over various types of infinite graphs have been studied, predominantly in the context of pushdown graphs~\cite{DBLP:conf/birthday/KupfermanPV10}.
For more general classes of graphs, a constraint-based approach~\cite{DBLP:conf/popl/BeyeneCPR14}, relying on constraint solvers such as Z3~\cite{DBLP:conf/tacas/MouraB08}, and various learning-based approaches have been proposed~\cite{DBLP:conf/wia/Neider10,DBLP:conf/tacas/NeiderT16} (we discuss these approaches in Section~\ref{sec:experiments}).
In the context of safety games over finite graphs, recent work~\cite{DBLP:conf/atva/Neider11} has demonstrated the ability of learning-based techniques to extract small controllers from precomputed controllers with a potentially large number of states.

Our learning framework is an extension of an earlier framework by Neider and Topcu~\cite{DBLP:conf/tacas/NeiderT16}.
Their work considers so-called rational safety games (defined in terms of finite automata) and proposes an automaton learning approach to infer winning strategies.
By contrast, we do not fix a specific representation of the game graph and only require that certain operations can be performed symbolically.
Many common formalisms such as finite automata, various types of decision diagrams, as well as formulas in the first-order theories of linear integer and real arithmetic satisfy these requirements.
However, we consider only finitely-branching game graphs, whereas Neider and Topcu also consider graphs with infinite branching.

The algorithm we design for learning decision trees builds on top of a learning algorithm recently proposed by Ezudheen et al.~\cite{DBLP:journals/pacmpl/EzudheenND0M18}, which learns from data in form of Horn clauses.
For this setting, other learning algorithms have been developed as well~\cite{DBLP:conf/tacas/ChampionC0S18,DBLP:conf/pldi/ZhuMJ18}.
We have chosen Ezudheen et al.'s algorithm specifically for its property to guarantee convergence to a solution in many practical scenarios.

%---------- Controller Synthesis and Safety Games ----------
% !Tex root=main.tex

\section{Controller Synthesis and Safety Games}
\label{sec:preliminaries}

We follow the game-theoretic approach to controller synthesis as popularized by McNaughton~\cite{DBLP:journals/apal/McNaughton93} and view the problem as an infinite-duration, two-player game on a directed graph.
Such games are played by two antagonistic players: \playera, who embodies the system, and \playerb, who embodies the environment.
In this setting, the type of specification dictates the type of game, and a winning strategy for \playera\ translates immediately into a controller that satisfies the given specification (we refer the reader to Grädel, Thomas, and Wilke~\cite{DBLP:conf/dagstuhl/2001automata} for a comprehensive discussion of this connection).
Since we are interested in synthesizing controllers from safety specifications, the remainder of this paper is concerned with so-called safety games.
However, before we introduce these types of games formally, let us first fix basic notations.

%---------- Basic Notations ----------
\paragraph*{Basic Notations}
Let $\mathbb B = \{ 0, 1 \}$ denote the set of Boolean values ($0$ representing $\mathit{false}$ and $1$ representing $\mathit{true}$), $\mathbb N$ the set of natural numbers, $\mathbb Z$ the set of integers, and $\mathbb R$ the set of real numbers.
Given a set $A$, we denote the set of all finite sequences of elements of $A$ by $A^\ast$ and the set of all infinite sequences by $A^\omega$.
Moreover, for a binary relation $R \subseteq X \times X$ and two sets $A, B \subseteq X$, the \emph{image of $A$ under $R$} is the set $R(A) = \{b \in B \mid \exists a \in A \colon (a, b) \in R \}$ and the \emph{preimage of $B$ under $R$} is the set $R^{-1}(B) = \{ a \in A \mid \exists b \in B \colon (a, b) \in R \}$.

%---------- Safety Games ----------
\paragraph*{Safety Games}
Our definitions and notations mainly follow those of Grädel, Thomas, and Wilke~\cite{DBLP:conf/dagstuhl/2001automata}, and we refer the reader to this textbook for further details on infinite-duration games.
Formally, a \emph{safety game} is a five-tuple $\mathcal G = (V_0, V_1, E, I, F)$ consisting of two disjoint sets $V_0$, $V_1$ of vertices controlled by \playera\ and \playerb, respectively (we denote their union by $V = V_0 \cup V_1$ and assume $V \neq \emptyset$), a directed edge relation $E \subseteq V \times V$, a nonempty set $I \subseteq V$ of \emph{initial vertices}, and a set $F \subseteq V$ of \emph{safe vertices}.
The directed graph $(V, E)$ is typically called \emph{game graph}.
In contrast to the classical setting, we do not restrict $V$ to be finite but allow even uncountable sets.
However, we do make the following two restrictions to the edge relation: we assume that
\begin{enumerate*}[label={(\arabic*)}]
	\item every vertex has at least one outgoing edge (i.e., $E(\{ v \}) \neq \emptyset$ for each $v \in V$), and
	\item $E(\{ v \})$ is finite for every $v \in V$, though not necessarily bounded.
\end{enumerate*}
Note that the first restriction is standard and simply avoids situations in which the game gets stuck.
The second restriction, on the other hand, is required by our learning framework and ensures that the data to learn from is always a finite object.

A safety game is played in rounds: initially, a token is placed on one of the initial vertices $v_0 \in I$; in each round, the player controlling the current vertex then moves the token to the next vertex along one of the outgoing edges.
This process of moving the token is repeated ad infinitum and results in an infinite sequence $\pi = v_0 v_1 \ldots \in V^\omega$ with $v_0 \in I$ and $(v_i, v_{i+1}) \in E$ for every $i \in \nat$, which is called a \emph{play}.
The winner of a play is determined by the winning condition $F$ in that a play $\pi = v_0 v_1 \ldots$ is \emph{winning for \playera} if $v_i \in F$ for every $i \in \nat$---otherwise it is \emph{winning for \playerb}.

In the framework of infinite games, synthesizing a controller amounts to computing a so-called winning strategy for \playera, which prescribes how \playera\ needs to move in order to win a play. % regardless of how \playerb\ plays.
Formally, a \emph{strategy for \playera} is a function $\sigma \colon V^\ast \times V_0 \to V$ such that $(v_n, \sigma(v_0 \ldots v_n)) \in E$ for every $v_0 \ldots v_n \in V^\ast V_0$.
A strategy is called \emph{winning} if every play that is \emph{played according to $\sigma$} (i.e., that satisfies $v_{n+1} = \sigma(v_0 \ldots v_n)$ for all $n \in \nat$ with $v_n \in V_0$) is winning for \playera.
It is well known that safety games permit memoryless winning strategies where the choice of the next vertex depends only on the vertex the play has currently reached.
Such a strategy can then easily be implemented as a controller: the controller tracks the current vertex of a play and chooses the next move according to the strategy.
Hence, the objective in the remainder of this paper is to compute a memoryless winning strategy for \playera.
We refer to this as \emph{solving a game}.

If the game graph underlying a safety game is finite, memoryless winning strategies can be computed in linear time using a simple fixed-point computation~\cite{DBLP:conf/dagstuhl/2001automata}.
For infinite game graphs, on the other hand, this is no longer an option as a fixed-point computation might not converge in finite time.
To overcome this problem, we propose a novel learning-based approach that learns a (memoryless) winning strategy via a proxy object named winning set.

%---------- Winning Sets ----------
\paragraph*{Winning Sets}
Intuitively, a winning set is a subset of the safe vertices that contains all initial vertices and is a trap for \playerb\ (i.e., \playera\ can force any play to stay inside this set regardless of how \playerb\ plays).
Formally, we define winning sets as follows.

\begin{definition}[Winning set] \label{def:winning-set}
Let $\mathcal G = (V_0, V_1, E, I, F)$ be a safety game.
A \emph{winning set} is a set $W \subseteq V$ satisfying
\begin{enumerate*}[label={(\arabic*)}]
	\item \label{itm:winning-set:1} $I \subseteq W$,
	\item \label{itm:winning-set:2} $W \subseteq F$,
	\item \label{itm:winning-set:3} $E(\{ v \}) \cap W \neq \emptyset$ for all $v \in W \cap V_0$ (\emph{existential closedness}), and
	\item \label{itm:winning-set:4} $E(\{ v \}) \subseteq W$ for all $v \in W \cap V_1$ (\emph{universal closedness}).
\end{enumerate*}
\end{definition}

A winning set $W$ immediately provides a winning strategy for \playera: starting in $I \subseteq W$, \playera\ simply moves to a (fixed) successor vertex inside $W$ whenever it is his turn (note that this is possible since $W$ is existentially closed).
As $W$ is also universally closed, a straightforward induction over the length of plays proves that every play that starts inside $I$ and is played according to this strategy stays inside $W$, no matter how \playerb\ plays.
Thus, \playera\ wins since $W \subseteq F$.

In the remainder, we encourage the reader to think of a winning set as a (symbolic) representation of a winning strategy.
The following example, inspired by robotic motion planning, illustrates the concept of winning sets.

\begin{example} \label{ex:games-and-winning-regions}
Consider a robot in an unbounded, one-dimensional grid world as depicted in Figure~\ref{fig:Example}.
Moreover, let us assume that the robot's position, indicated by an ``\textsf{x}'' in Figure~\ref{fig:Example}, can be modeled as a single real-valued coordinate $x \in \mathbb R$, allowing for an uncountable number of positions.

\begin{figure}[t]
\centering
\begin{tikzpicture} 
  \tikzstyle{Unsafe} = [fill = white] 
  \tikzstyle{Safe} = [fill=white] 
  \tikzstyle{Initial} = [pattern=north east lines, pattern color=black!70]
  \tikzstyle{InitialNode} = [Initial,
    minimum width=0.8cm, minimum height=0.8cm, node distance=0.8cm, text width=0.65cm,align=right,text height=0.65cm]
  \tikzstyle{Node} = [Unsafe,
    minimum width=0.8cm, minimum height=0.8cm, node distance=0.8cm, text width=0.65cm,align=right,text height=0.65cm]
  \tikzstyle{SafeNode} = [Node, Safe]
  \node [label={[xshift=-0.5cm]-2}, name=n0, Node]{0} (0,0);
  \node[label={[xshift=-0.5cm]-1},name=n1, right of=n0, Node]{1} ;
   % \fill [ (1.2,-0.4) rectangle (2.0,0.4);
  \node[label={[xshift=-0.5cm]0},name=n2, right of=n1, InitialNode]{0} ;

  \node[label={[xshift=-0.5cm]1},name=n3, right of=n2, SafeNode]{1};
  \node[label={[xshift=-0.5cm]2},name=n4, right of=n3, SafeNode]{0} ;
  \node [label={[xshift= 0.2cm,yshift=0.32cm]3},name=n4,right of =n3]{};
  \node [label={[xshift= -0.1cm,yshift=0.4cm]...},right of =n4]{};
  \node [label={[xshift= -0.1cm,yshift=0.4cm]...},left of =n0]{};

  \draw[gray!50] (2.4, .4) -- node[anchor=center, fill=white, inner sep=1pt, font=\large\sffamily, text=black] {x} ++(0, -.8);

  \draw [thick, black, dotted] (-0.4,0.4)  -- (-0.4,-0.4);
  \draw [thick, black, dotted] (0.4,0.4)  -- (0.4,-0.4);
  \draw [thick, black, dotted] (1.2,0.4)  -- (1.2,-0.4);
  \draw [thick, black, line width = 0.7mm] (1.2,0.41)  -- (1.2,-0.41);
  \draw [thick, black, dotted] (2.0,0.4)  -- (2.0,-0.4);
  \draw [thick, black, dotted] (2.8,0.4)  -- (2.8,-0.4);
  \draw [thick, black, dotted] (3.6,0.4)  -- (3.6,-0.4);
  \draw [thin, black] (-0.4,0.4)  -- (3.6,0.4);
  \draw [thick, black, line width = 0.7mm] (1.16,0.4)  -- (3.65,0.4);
  \draw [thick, black, line width = 0.7mm] (1.16,-0.4)  -- (3.65,-0.4);
  \draw [thin, black] (-0.4,-0.4)  -- (3.6,-0.4);
  \draw [thin, black,  dashed, line width = 0.7mm] (3.6,0.4)  -- (4.4,0.4);
  \draw [thin, black,  dashed, line width = 0.7mm] (3.6,-0.4)  -- (4.4,-0.4);
  \draw [thin, black,  dashed] (-0.4,0.4)  -- (-1.2,0.4);
  \draw [thin, black,  dashed] (-0.4,-0.4)  -- (-1.2,-0.4);
  
\end{tikzpicture}
\caption[A robot on the one-dimensional grid word described in Example~\ref{ex:games-and-winning-regions}]{A robot on the one-dimensional grid word described in Example~\ref{ex:games-and-winning-regions}. The division of the world into intervals is indicated by dotted ``boxes''. The number in the lower-right corner of each box displays the player having control over the robot in this interval. Safe vertices are surrounded by a bold line \tikz \draw[very thick] (0, 0) rectangle +(1.5ex, 1.5ex);, initial vertices are indicated by a diagonal pattern \tikz \draw[black!70, pattern=north east lines, pattern color=black!70] (0, 0) rectangle +(1.5ex, 1.5ex);.} \label{fig:Example}
\end{figure}
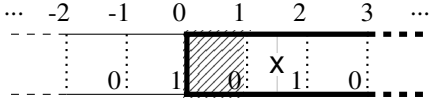

The robot's movement is controlled by two players, the system (\playera) and the environment (\playerb).
To decide which player is currently in control of the robot, we divide the world in infinitely many intervals $[n, n+1) \subset \mathbb R$ with $n \in \mathbb Z$: if the robot resides inside an interval $[n, n+1)$ with $n$ even, then \playera\ is in control; otherwise, \playerb\ is in control.
In every turn, we allow the system or the environment to move the robot one unit to the left (i.e., decreasing its position $x$ by one) or one unit to the right (i.e., increasing its position $x$ by one).
Note that this implies that the system and the environment are taking turns controlling the robot. 
We initially place the robot at an arbitrary position $x \in [0, 1)$ and define all positions $x \geq 0$ to be safe.
The objective of the system is to stay inside these safe positions at all times.

We can model the above situation as a safety game $\mathcal G = (V_0, V_1, E, I, F)$ as follows.
The set of \playera\ vertices is $V_0 = \bigcup_{n \in \mathbb Z} [2n, 2n+1)$, while the set of \playerb\ vertices is $V_1 = \bigcup_{n \in \mathbb Z} [2n-1, 2n)$; note that $V = \mathbb R$.
The set of safe vertices is $F = \{ x \in \mathbb R \mid x \geq 0 \}$ and the set of initial vertices is $I = [0, 1)$.
Finally, the edge relation is defined as $E = \{ (x, x+1) \in \mathbb R \times \mathbb R \mid x \in \mathbb R \} \cup \{ (x, x-1) \in \mathbb R \times \mathbb R \mid x \in \mathbb R \}$.

An example of a winning set is $W = [0, 3)$.
It clearly contains the set $I = [0, 1)$ of initial vertices.
Furthermore, if the position of the robot is in the interval $[0, 1)$ or $[2, 3)$, \playera\ can stay inside $W$ by moving the robot into the interval $[1, 2)$.
Similarly, if the robot is in the interval $[1, 2)$, \playerb\ has no choice but to move the robot either into the interval $[0,1)$ or into $[2,3)$, thus staying inside $W$.
Since $W$ is also contained in $F = \{ x \in \mathbb R \mid x \geq 0 \}$, it is in fact a winning set. 
As sketched above, a winning strategy is easy to derive. \exampleend
\end{example}

Finally, let us note that safety games as defined above subsume games on finite graphs, but one needs to choose a suitable symbolic representation in order to handle \mbox{(un-)countably} infinite game graphs.
We comment on this in detail after having introduced our machine learning framework in the next section.

%---------- A Machine Learning Framework for Synthesizing Safety Controllers ----------
% !Tex root=main.tex

\section{A Machine Learning Framework for Synthesizing Safety Controllers}
\label{sec:framework}

We now propose a machine learning framework for learning winning sets in safety games and, thus, reactive safety controllers.
Our framework is a generalization of earlier work by Neider and Topcu~\cite{DBLP:conf/tacas/NeiderT16}, which encodes (countably) infinite game graphs using finite automata and uses automaton learning to learn winning sets.
By contrast, the framework proposed here allows for game graphs with uncountably many vertices and is not restricted to a symbolic representation in terms of finite automata.

As illustrated in Figure~\ref{fig:learning-framework}, the learning takes place in a counterexample-guided feedback loop (CEGIS)~\cite{DBLP:conf/aplas/Solar-Lezama09} with two entities: a \emph{teacher}, who has knowledge about the safety game, and a \emph{learner} (or learning algorithm), whose objective is to learn a winning set, but who is agnostic to the game.
In every iteration of the loop, the learner conjectures a set $H \subseteq V$, called \emph{hypothesis}, based on the information about the game it has accumulated so far.
Then, the teacher checks whether this set $H$ is in fact a winning set---queries of this type are often called equivalence or correctness queries.
Although the teacher does not know a winning set (the task is to learn one after all), it can verify whether the hypothesis is one by checking Conditions~\ref{itm:winning-set:1} to \ref{itm:winning-set:4} of Definition~\ref{def:winning-set}.
If the hypothesis satisfies these conditions, then $H$ is a winning set and the learning stops.
If this is not the case, the teacher replies with a counterexample that witnesses the violation of one of these conditions.
Then, the feedback loop continues until a winning set has been found.
The definition below formalizes the concept of counterexamples and fixes the communication between the teacher and the learner.

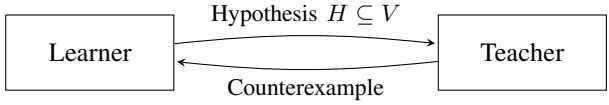
\begin{figure}[!t]
	\centering

	\begin{tikzpicture}
		\node[draw, minimum height=10mm, text width=20mm, align=center] (L) at (0, 0) {Learner};
		\node[draw, minimum height=10mm, text width=20mm, align=center] (T) at (5.75, 0) {Teacher};
		\draw[->] (L) to[in=174, out=6] node[font=\small] {Hypothesis $H \subseteq V$} (T);
		\draw[->] (T) to[in=-6, out=-174] node[font=\small] {Counterexample} (L);
	\end{tikzpicture}

	\caption{Learning framework for synthesizing safety controllers} \label{fig:learning-framework}
\end{figure}

\begin{definition}[Teacher for Safety Games] \label{def:teacher}
Let $\mathcal G = (V_0, V_1, E, I, F)$ be a safety game.
Given a hypothesis $H \subseteq V$, the teacher replies as follows (whereby the order in which the checks are performed is arbitrary):
\begin{enumerate}
	\item If $I \not\subseteq H$, then the teacher returns a \emph{positive counterexample} $v \in I \setminus H$.
	\item If $H \not\subseteq F$, then the teacher returns a \emph{negative counterexample} $v \in H \setminus F$.
	\item If there exists a $v \in H \cap V_0$ with $E(\{ v \}) \cap H = \emptyset$, then the teacher returns an \emph{existential counterexample} $v \rightarrow (v_1 \lor \ldots \lor v_n)$ with $\{ v_1, \ldots, v_n \} = E(\{ v \})$.
	\item If there exists a $v \in H \cap V_1$ with $E(\{ v \}) \not\subseteq H$, then the teacher returns a \emph{universal counterexample} $v \rightarrow (v_1 \land \ldots \land v_n)$ with $\{ v_1, \ldots, v_n \} = E(\{ v \})$.
\end{enumerate}
If $H$ passes all four checks, then the teacher returns ``yes''.
\end{definition}

Each case of Definition~\ref{def:teacher} corresponds to one condition of Definition~\ref{def:winning-set}.
Thus, it is not hard to verify that a given hypothesis is in fact a winning set if the teacher replies ``yes'' (as it satisfies Definition~\ref{def:winning-set}).
Counterexamples, on the other hand, witness the violation of one of these conditions and guide the learner towards a winning set by communicating exactly why the hypothesis is incorrect.
For instance, the meaning of a positive counterexample is that any future hypothesis needs to include this vertex (as it is initial), whereas a negative counterexample must be excluded (as it is not a safe vertex).
An existential counterexample $v \rightarrow (v_1 \lor \ldots \lor v_n)$ signals that the hypothesis is not existentially closed and requires that if a future hypothesis contains $v$, it also needs to contains an least one vertex of the vertices $v_1, \ldots, v_n$.
Similarly, a universal counterexample $v \rightarrow (v_1 \land \ldots \land v_n)$ signals that the hypothesis is not universally closed and requires that if a future hypothesis contains $v$, it needs to contains all vertices $v_1, \ldots, v_n$.
Note that existential and universal counterexamples are always finite objects since we assume $E(\{ v \})$ to be finite for every $v \in V$.

Let us illustrate the overall learning process with an example.

\begin{example} \label{ex:learning}
We continue Example~\ref{ex:games-and-winning-regions}.
All hypotheses produced in the course of the learning process are depicted in Figure~\ref{fig:LearningLoop}.
Gray shaded areas indicate vertices in the hypothesis, safe vertices are surrounded by a bold line, and initial vertices are indicated by diagonal lines.

\begin{figure}[t!h]
	\centering
\newcommand{\varname}{0.25}
	\begin{subfigure}[b]{0.4\textwidth}
		\centering

		\begin{tikzpicture} 
		  \tikzstyle{Initial} = [pattern=north east lines, pattern color=black!70]
  \tikzstyle{InitialNode} = [Initial, minimum width=0.6cm, minimum height=0.6cm, node distance=0.6cm, text width=0.4cm,align=right,text height=0.4cm]
			\tikzstyle{Unsafe} =[fill = white] 
			\tikzstyle{Safe} = [fill=white] 
			\tikzstyle{Node} = [Unsafe, minimum width=0.6cm, minimum height=0.6cm, node distance=0.6cm, text width=0.4cm,align=right,text height=0.4cm]
			\tikzstyle{SafeNode} = [Node, Safe]
			
			\node [label={[xshift=-0.3cm]-2},label={[xshift=-0.1cm]}, name=n0, Node]{0};
			\node[label={[xshift=-0.3cm]-1},name=n1, right of=n0, Node]{1} ;
			\node[label={[xshift=-0.3cm]0},name=n2, right of=n1, InitialNode]{0} ;
			\node[label={[xshift=-0.3cm]1},name=n3, right of=n2, Node]{1};
			\node[label={[xshift=-0.3cm]2},name=n4, right of=n3, Node]{0} ;
			\node [label={[xshift= -0.7cm,yshift=0.2cm]3},right of =n4]{};
			\node [label={[xshift= -0.3cm,yshift=0.2cm]...},right of =n4]{};
			\node [label={[xshift= 0.3cm,yshift=0.2cm]...},left of =n0]{};
			\draw [thin, black] (-0.3,0.3)  -- (2.7,0.3);
			\draw [thin, black] (-0.3,-0.3)  -- (2.7,-0.3);
			
			\draw [thick, black, dotted] (-0.3,0.3)  -- (-0.3,-0.3);
			\draw [thick, black, dotted] (0.3,0.3)  -- (0.3,-0.3);
			\draw [thick, black, dotted] (0.9,0.3)  -- (0.9,-0.3);
			\draw [thick, black, dotted] (2.1,0.3)  -- (2.1,-0.3);
			\draw [thick, black, dotted] (1.5,0.3)  -- (1.5,-0.3);
			\draw [thick, black, dotted] (2.7,0.3)  -- (2.7,-0.3);		
			
			\draw [thin, black,  dashed, line width = 0.7mm] (2.7,0.3)  -- (3.3,0.3);
			\draw [thin, black,  dashed, line width = 0.7mm] (2.7,-0.3)  -- (3.3,-0.3);
			\draw [thin, black,  dashed] (-0.3,0.3)  -- (-0.9,0.3);
			\draw [thin, black,  dashed] (-0.3,-0.3)  -- (-0.9,-0.3);
			\draw [thick, black, line width = 0.7mm] (0.87,0.3)  -- (2.77,0.3);
			\draw [thick, black, line width = 0.7mm] (0.87,-0.3)  -- (2.77,-0.3); 
			\draw [thick, black, line width = 0.7mm] (0.9,0.3)  -- (0.9,-0.3);
		\end{tikzpicture}
		\subcaption{$H_1 = \emptyset$} \label{fig:PositiveCounterexample}

	\end{subfigure}

	\begin{subfigure}[b]{0.25\textwidth}

		\begin{tikzpicture} 
\tikzstyle{Initial} = [pattern=north east lines, pattern color=black!70]
  \tikzstyle{InitialNode} = [Initial, minimum width=2* \varname cm, minimum height=2*\varname cm, node distance=\varname cm, text width=0.3cm,align=right,text height=0.3cm]
			\tikzstyle{Unsafe} = [fill = gray!40!white] 
			\tikzstyle{Safe} = [fill = gray!40!white] 
			\tikzstyle{Node} = [ Unsafe, minimum width=2*\varname cm, minimum height=2*\varname cm, node distance=2*\varname cm, text width=0.3cm,align=right,text height=0.3cm]
			\tikzstyle{SafeNode} = [Node, Safe]

			\node [label={[xshift=-0.25cm]-2},xshift=-0.25cm, name=n0, Node]{0};
			\node[label={[xshift=-0.25cm]-1},name=n1, right of=n0, Node]{1} ;
			\fill [gray!40!white] (2*\varname,\varname) rectangle +(2*\varname,-2*\varname);			
			\node[label={[xshift=-0.25cm]0},xshift=0.25cm,name=n2, right of=n1, InitialNode]{0} ;
			\node[label={[xshift=-0.25cm]1},name=n3, right of=n2, SafeNode]{1};

			\shade[left color=gray!40!white, right color=gray!10!white] (8*\varname,-1*\varname) rectangle +(2*\varname,2*\varname);
			\shade[right color=gray!40!white, left color=gray!10!white] (-2*\varname,-1*\varname) rectangle +(-2*\varname,2*\varname);
			\node[label={[xshift=-0.25cm]2},name=n4, right of=n3, SafeNode]{0} ;
			\node [label={[xshift= -0.75cm,yshift=0.15cm]3},right of =n4]{};
			\node [label={[xshift= -0.4cm,yshift=0.2cm]...},right of =n4]{};
			\node [label={[xshift= 0.4cm,yshift=0.2cm]...},left of =n0]{};		 
			\draw [thin, black] (-2*\varname,\varname)  -- (8*\varname,\varname);
			\draw [thin, black] (-2*\varname,-\varname)  -- (8*\varname,-1*\varname);

			\draw [thick, black, dotted] (-2*\varname,\varname)  -- (-0.5,-0.25);
			\draw [thick, black, dotted] (2*\varname,\varname)  -- (2*\varname,-1*\varname);
			\draw [thick, black, dotted] (4*\varname,\varname)  -- (4*\varname,-1*\varname);
			\draw [thick, black, dotted] (8*\varname,\varname)  -- (8*\varname,-1*\varname);
			\draw [thick, black, dotted] (6*\varname,\varname)  -- (6*\varname,-1*\varname);
			\draw [thick, black, dotted] (0.0,\varname)  -- (0.0,-1*\varname);	
			
			\draw [thin, black,  dashed, line width = 0.7mm] (8*\varname,\varname)  -- (10*\varname,\varname);
			\draw [thin, black,  dashed, line width = 0.7mm] (8*\varname,-1*\varname)  -- (10*\varname,-1*\varname);
			\draw [thin, black,  dashed] (-2*\varname,\varname)  -- (-4*\varname,\varname);
			\draw [thin, black,  dashed] (-2*\varname,-\varname)  -- (-4*\varname,-\varname);
			\draw [thick, black, line width = 0.7mm] (2*\varname,\varname)  -- (8*\varname,\varname);
			\draw [thick, black, line width = 0.7mm] (2*\varname,-\varname)  -- (8*\varname,-\varname); 
			\draw [thick, black, line width = 0.7mm] (2*\varname,\varname)  -- (2*\varname,-\varname);
		\end{tikzpicture}
		\subcaption{$H_2 = V\;\;\;\;$}
		 \label{fig:NegativeCounterexample}
	\end{subfigure}
	\hskip -0.05\linewidth
	\begin{subfigure}[b]{0.25\textwidth}
		\begin{tikzpicture}
\tikzstyle{Initial} = [pattern=north east lines, pattern color=black!70]
  \tikzstyle{InitialNode} = [Initial, minimum width=0.5cm, minimum height=0.5cm, node distance=0.25cm, text width=0.3cm,align=right,text height=0.3cm]
			\tikzstyle{Unsafe} = [fill = white]] 
			\tikzstyle{Safe} = [fill = gray!40!white] 
			\tikzstyle{Node} = [ Unsafe, minimum width=0.5cm, minimum height=0.5cm, node distance=0.5cm, text width=0.3cm,align=right,text height=0.3cm]
			\tikzstyle{SafeNode} = [Node, Safe]

			\node [label={[xshift=-0.25cm]-2},xshift=-0.25cm, name=n0, Node]{0};
			\node[label={[xshift=-0.25cm]-1},name=n1, right of=n0, Node]{1} ;
			\fill [gray!40!white] (0.5,0.25) rectangle +(0.5,-0.5);			
			\node[label={[xshift=-0.25cm]0},xshift=0.25cm,name=n2, right of=n1, InitialNode]{0} ;
			\node[label={[xshift=-0.25cm]1},name=n3, right of=n2, Node]{1};

			\node[label={[xshift=-0.25cm]2},name=n4, right of=n3, Node]{0} ;
			\node [label={[xshift= -0.75cm,yshift=0.15cm]3},right of =n4]{};
			\node [label={[xshift= -0.4cm,yshift=0.2cm]...},right of =n4]{};
			\node [label={[xshift= 0.4cm,yshift=0.2cm]...},left of =n0]{};	 
			\draw [thin, black] (-0.5,0.25)  -- (2.0,0.25);
			\draw [thin, black] (-0.5,-0.25)  -- (2.0,-0.25);

			\draw [thick, black, dotted] (-0.5,0.25)  -- (-0.5,-0.25);
			\draw [thick, black, dotted] (0.5,0.25)  -- (0.5,-0.25);
			\draw [thick, black, dotted] (1.0,0.25)  -- (1.0,-0.25);
			\draw [thick, black, dotted] (2.0,0.25)  -- (2.0,-0.25);
			\draw [thick, black, dotted] (1.5,0.25)  -- (1.5,-0.25);
			\draw [thick, black, dotted] (0.0,0.25)  -- (0.0,-0.25);	
			
			\draw [thin, black,  dashed, line width = 0.7mm] (2.0,0.25)  -- (2.5,0.25);
			\draw [thin, black,  dashed, line width = 0.7mm] (2.0,-0.25)  -- (2.5,-0.25);
			\draw [thin, black,  dashed] (-0.5,0.25)  -- (-1.0,0.25);
			\draw [thin, black,  dashed] (-0.5,-0.25)  -- (-1.0,-0.25);
			\draw [thick, black, line width = 0.7mm] (0.5,0.25)  -- (2.0,0.25);
			\draw [thick, black, line width = 0.7mm] (0.5,-0.25)  -- (2.0,-0.25); 
			\draw [thick, black, line width = 0.7mm] (0.5,0.25)  -- (0.5,-0.25);
		\end{tikzpicture}
		\subcaption{$H_3 = [0, 1)$} \label{fig:ExistentialCounterexample}
	\end{subfigure}%

	%\smallskip
	\begin{subfigure}[b]{0.25\textwidth}

		\begin{tikzpicture} 
\tikzstyle{Initial} = [pattern=north east lines, pattern color=black!70]
  \tikzstyle{InitialNode} = [Initial, minimum width=0.5cm, minimum height=0.5cm, node distance=0.25cm, text width=0.3cm,align=right,text height=0.3cm]
			\tikzstyle{Unsafe} = [fill = white] 
			\tikzstyle{Safe} = [fill = gray!40!white] 
			\tikzstyle{Node} = [ Unsafe, minimum width=0.5cm, minimum height=0.5cm, node distance=0.5cm, text width=0.3cm,align=right,text height=0.3cm]
			\tikzstyle{SafeNode} = [Node, Safe]

			\node [label={[xshift=-0.25cm]-2},xshift=-0.25cm, name=n0, Node]{0};
			\node[label={[xshift=-0.25cm]-1},name=n1, right of=n0, Node]{1} ;
			\fill [gray!40!white] (0.5,0.25) rectangle +(0.5,-0.5);			
			\node[label={[xshift=-0.25cm]0},xshift=0.25cm,name=n2, right of=n1, InitialNode]{0} ;
			\node[label={[xshift=-0.25cm]1},name=n3, right of=n2, SafeNode]{1};

			\node[label={[xshift=-0.25cm]2},name=n4, right of=n3, Node]{0} ;
			\node [label={[xshift= -0.75cm,yshift=0.15cm]3},right of =n4]{};
			\node [label={[xshift= -0.4cm,yshift=0.2cm]...},right of =n4]{};
			\node [label={[xshift= 0.4cm,yshift=0.2cm]...},left of =n0]{};	
			\draw [thin, black] (-0.5,0.25)  -- (2.0,0.25);
			\draw [thin, black] (-0.5,-0.25)  -- (2.0,-0.25);

			\draw [thick, black, dotted] (-0.5,0.25)  -- (-0.5,-0.25);
			\draw [thick, black, dotted] (0.5,0.25)  -- (0.5,-0.25);
			\draw [thick, black, dotted] (1.0,0.25)  -- (1.0,-0.25);
			\draw [thick, black, dotted] (2.0,0.25)  -- (2.0,-0.25);
			\draw [thick, black, dotted] (1.5,0.25)  -- (1.5,-0.25);
			\draw [thick, black, dotted] (0.0,0.25)  -- (0.0,-0.25);	
			
			\draw [thin, black,  dashed, line width = 0.7mm] (2.0,0.25)  -- (2.5,0.25);
			\draw [thin, black,  dashed, line width = 0.7mm] (2.0,-0.25)  -- (2.5,-0.25);
			\draw [thin, black,  dashed] (-0.5,0.25)  -- (-1.0,0.25);
			\draw [thin, black,  dashed] (-0.5,-0.25)  -- (-1.0,-0.25);
			\draw [thick, black, line width = 0.7mm] (0.5,0.25)  -- (2.0,0.25);
			\draw [thick, black, line width = 0.7mm] (0.5,-0.25)  -- (2.0,-0.25); 
			\draw [thick, black, line width = 0.7mm] (0.5,0.25)  -- (0.5,-0.25);
		\end{tikzpicture}
		\subcaption{$H_4 = [0, 2)$} \label{fig:UniversalCounterexample}
	\end{subfigure}
	\hskip -0.05\linewidth
	\begin{subfigure}[b]{0.25\textwidth}
		\begin{tikzpicture} 
\tikzstyle{Initial} = [pattern=north east lines, pattern color=black!70]
  \tikzstyle{InitialNode} = [Initial, minimum width=0.5cm, minimum height=0.5cm, node distance=0.25cm, text width=0.3cm,align=right,text height=0.3cm]
			\tikzstyle{Unsafe} = [fill = white] 
			\tikzstyle{Safe} = [fill = gray!40!white] 
			\tikzstyle{Node} = [ Unsafe, minimum width=0.5cm, minimum height=0.5cm, node distance=0.5cm, text width=0.3cm,align=right,text height=0.3cm]
			\tikzstyle{SafeNode} = [Node, Safe]

			\node [label={[xshift=-0.25cm]-2},xshift=-0.25cm, name=n0, Node]{0};
			\node[label={[xshift=-0.25cm]-1},name=n1, right of=n0, Node]{1} ;
			\fill [gray!40!white] (0.5,0.25) rectangle +(0.5,-0.5);			
			\node[label={[xshift=-0.25cm]0},xshift=0.25cm,name=n2, right of=n1, InitialNode]{0} ;
			\node[label={[xshift=-0.25cm]1},name=n3, right of=n2, SafeNode]{1};
			
			\node[label={[xshift=-0.25cm]2},name=n4, right of=n3, SafeNode]{0} ;
			\node [label={[xshift= -0.75cm,yshift=0.15cm]3},right of =n4]{};
			\node [label={[xshift= -0.4cm,yshift=0.2cm]...},right of =n4]{};
			\node [label={[xshift= 0.4cm,yshift=0.2cm]...},left of =n0]{};	
		
						\draw [thin, black] (-0.5,0.25)  -- (2.0,0.25);
			\draw [thin, black] (-0.5,-0.25)  -- (2.0,-0.25);

			\draw [thick, black, dotted] (-0.5,0.25)  -- (-0.5,-0.25);
			\draw [thick, black, dotted] (0.5,0.25)  -- (0.5,-0.25);
			\draw [thick, black, dotted] (1.0,0.25)  -- (1.0,-0.25);
			\draw [thick, black, dotted] (2.0,0.25)  -- (2.0,-0.25);
			\draw [thick, black, dotted] (1.5,0.25)  -- (1.5,-0.25);
			\draw [thick, black, dotted] (0.0,0.25)  -- (0.0,-0.25);	
			
			\draw [thin, black,  dashed, line width = 0.7mm] (2.0,0.25)  -- (2.5,0.25);
			\draw [thin, black,  dashed, line width = 0.7mm] (2.0,-0.25)  -- (2.5,-0.25);
			\draw [thin, black,  dashed] (-0.5,0.25)  -- (-1.0,0.25);
			\draw [thin, black,  dashed] (-0.5,-0.25)  -- (-1.0,-0.25);
			\draw [thick, black, line width = 0.7mm] (0.5,0.25)  -- (2.0,0.25);
			\draw [thick, black, line width = 0.7mm] (0.5,-0.25)  -- (2.0,-0.25); 
			\draw [thick, black, line width = 0.7mm] (0.5,0.25)  -- (0.5,-0.25);
		\end{tikzpicture}
	
		\subcaption{$H_5 = [0, 3)$} \label{fig:WinningSet}
	\end{subfigure}%
	
	\caption{Hypotheses produced in the course of Example~\ref{ex:learning}.} \label{fig:LearningLoop}
\end{figure}
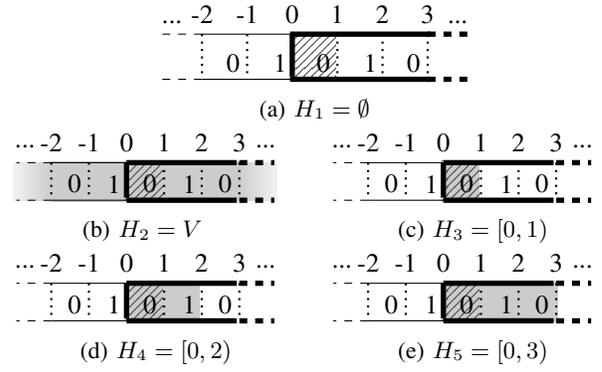

Let us assume that the learner proposes the hypothesis $H_1 = \emptyset$ in the first iteration of the loop, which is shown in Figure~\ref{fig:PositiveCounterexample}.
Since this hypothesis does not include any initial vertex, it does not satisfy Condition~\ref{itm:winning-set:1} of Definition~\ref{def:winning-set}.
Thus, the teacher returns a positive counterexample, say $0 \in I \setminus H$.

Next, suppose that the learner proposes the hypothesis $H_2 = V$, shown in Figure~\ref{fig:NegativeCounterexample}.
This hypothesis is consistent with the positive counterexample.
However, it also includes unsafe vertices and, thus, does not satisfy Condition~\ref{itm:winning-set:2} of Definition~\ref{def:winning-set}.
Hence, the teacher replies with a negative counterexample, say $-1 \in H \setminus F$.

Let us now assume that the learner conjectures $H_3 = [0,1)$, depicted in Figure~\ref{fig:ExistentialCounterexample}.
This conjecture is consistent with both the positive and the negative counterexample.
However, it is not existentially closed because \playera\ has no choice but to move the robot outside of $H_3$.
Therefore, the teacher replies with an existential counterexample, say $0 \rightarrow (-1 \lor 1)$.

Suppose now that the learner proposes the hypothesis $H_4 = [0, 2)$ in the fourth iteration, shown in Figure \ref{fig:UniversalCounterexample}.
Although this conjecture is consistent with all counterexamples received so far, it is not universally closed because \playerb\ can move the robot into the interval $[2, 3)$, which is not included in $H_4$.
Thus, the teacher returns a universal counterexample, say $2 \rightarrow (1 \land 3)$.

Finally, the learner proposes the hypothesis $H_5 = [0, 3)$, depicted in Figure~\ref{fig:WinningSet}.
This hypothesis satisfies all conditions of Definition~\ref{def:teacher}.
Thus, $H_5$ is a winning set and the learning terminates. \exampleend
\end{example}

Our learning framework is straightforward to implement if the underlying game graph is finite, in which case the teacher can be built on top of an explicit representation of the game.
However, if the underlying game graph becomes too large or is infinite, one has to choose a suitable representation for sets of vertices and the edge relation that allows performing operations on the graph symbolically.
More precisely, the chosen symbolic representation must feature Boolean operations (i.e., union, intersection, and complementation), and the image $E(A)$ and preimage $E^{-1}(A)$ of symbolically represented sets $A \subseteq V$ need to be computable.
Moreover, the emptiness problem (i.e., ``given a set $A$, decide whether $A = \emptyset$'') needs to be decidable, and it must be possible to extract an element from $A$ if it is nonempty.
Examples of such symbolic representations include many common formalisms such as finite automata, various types of decision diagrams, and first-order formulas in linear integer arithmetic and real arithmetic.

Furthermore, hypotheses must also be expressible in the chosen symbolic formalism.
However, to build efficient learning algorithms, it is often necessary to restrict the class of hypotheses---typically called the \emph{hypothesis space} and denoted by $\mathcal H$---even further.
Examples of such restricted hypothesis spaces are conjunctive formulas (e.g., as used by the popular Houdini algorithm~\cite{DBLP:conf/fm/FlanaganL01}) or decision trees, which are common, for instance, in learning-based software verification~\cite{DBLP:conf/tacas/ChampionC0S18,DBLP:conf/cav/0001LMN14,DBLP:conf/popl/0001NMR16,DBLP:conf/pldi/ZhuMJ18} and also used in this work.
Note that it might happen that a winning set exists, though it cannot be expressed as a hypothesis in the hypothesis space.
Thus, the choice of the hypothesis space and, hence, the learning algorithm needs to be made carefully in order for the learning to succeed.

A second important property of learning algorithms is what we call ``consistency''.
To make this notion precise, let us assume that the learner accumulates counterexamples in a so-called \emph{game sample} $\mathcal S_G = (\Pos, \Neg, \Ex, \Un)$ consisting of a finite set $\Pos$ of positive counterexamples, a finite set $\Neg$ of negative counterexample, a finite set $\Ex$ of existential counterexamples, and a finite set $\Un$ of universal counterexamples.
Then, we say that a hypothesis $H \subseteq V$ is \emph{consistent} with a game sample $\mathcal S_G = (\Pos, \Neg, \Ex, \Un)$ if
\begin{enumerate*}[label={(\arabic*)}]
	\item $v \in H$ for each $v \in \Pos$,
	\item $v \notin H$ for each $v \in \Neg$,
	\item $v \in H$ implies $\{ v_1, \ldots, v_n \} \cap H \neq \emptyset$ for each $v \rightarrow (v_1 \lor \ldots \lor v_n) \in \Ex$, and 
	\item $v \in H$ implies $\{ v_1, \ldots, v_n \} \subseteq H$ for each $v \rightarrow (v_1 \land \ldots \land v_n) \in \Un$.
\end{enumerate*}
Moreover, we call a learner \emph{consistent} if it always produces a consistent hypothesis.
Consistency is an important property as it prevents the learner from making the same mistake twice and ensures progress towards a winning set.

In fact, the notion of consistency allows us to show that our framework is sound in the sense that any consistent learner learns a winning set in the limit if one exists in the chosen hypothesis space.
This is formalized in the next theorem.

\begin{theorem} \label{thm:framework-sound}
Let a teacher for a safety game (as described in Definition~\ref{def:teacher}) and a consistent learner over an hypothesis space $\mathcal H$ be given.
If there exists a winning set expressible as a hypothesis in $\mathcal H$, then there exists an ordinal $\alpha \in \mathbb O$, where $\mathbb O$ denotes the class of all ordinals, such that the learner proposes a winning set after at most $\alpha$ iterations.
\end{theorem}

Theorem~\ref{thm:framework-sound} is a consequence of the fact that our learning framework is in instance of an \emph{abstract learning framework for synthesis (ALF)}, as introduced by Löding et al.~\cite{DBLP:conf/tacas/LodingMN16}.
The proof roughly proceeds as follows.
Since the teacher of Definition~\ref{def:teacher} allows ``progress'' (i.e., every counterexample refutes the current
hypothesis) and we assume the learner to be consistent, the learner never conjectures the same hypothesis twice.
In the worst case, the learner will have exhausted all incorrect hypotheses after $\alpha$ iterations for an ordinal $\alpha \in \mathbb O$ with cardinality less or equal to $|\mathcal H|$.
Since the teacher is also ``honest'' (i.e., it does not return spurious counterexamples), the learner necessarily produces a winning set in the subsequent iteration if one exists. 

Finally, let us point out that the safety games as defined in Section~\ref{sec:preliminaries} are very general and even allow encoding computations of Turing machines.
Consequently, determining the winner of such safety games is undecidable in general, and any algorithm  for  computing winning sets can be a semi-algorithm at best (i.e., an algorithm that, on termination, gives the correct answer but does not guarantee to halt).
The algorithm we design in the next section is of this kind.

%---------- Decision Tree Learner ----------
% !Tex root=main.tex

\section{Learning Decision Trees from Game Samples}
\label{sec:dt-learner}

We now fix the hypothesis space $\mathcal H$ to be the class of all decision trees (as defined shortly) and describe an algorithm to learn consistent decision trees from game samples.
Our algorithm builds on top of a learning algorithm recently proposed by Ezudheen et al.~\cite{DBLP:journals/pacmpl/EzudheenND0M18} in the context of software verification.
To ease the presentation in this section, we abstract from the setting of infinite games and assume that the data to learn from is taken from an abstract domain $\mathcal D$, whose elements we call \emph{data points}.
We encourage the reader to think of data points and vertices as synonyms and define concepts such as game samples and consistency analogously for data points.

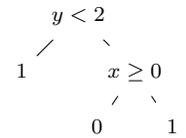
\begin{wrapfigure}[9]{r}{34mm}
\setlength\abovecaptionskip{1\baselineskip}
	\centering
	\vskip -0.5\baselineskip
	\vskip -.3333em % Ticks inner sep
	\begin{tikzpicture}
		\begin{scope}[every node/.append style={font=\strut\footnotesize}]
			\node (0) at (0, 0) {$y < 2$};
			\node (1) at (-.75, -.75) {$1$};
			\node (2) at (.75, -.75) {$x \geq 0$};
			\node (3) at (.25, -1.5) {$0$};
			\node (4) at (1.25, -1.5) {$1$};
		\end{scope}

		\draw (0) -- (1) (0) -- (2) (2) -- (3) (2) -- (4);
	\end{tikzpicture}
	\vskip -1em
	\caption{A decision tree over $\mathcal D = \mathbb R$ and $\mathcal P = \{ x \geq 0, y < 2 \}$.} \label{fig:decision-tree}
\end{wrapfigure}
An example of a decision tree is shown in Figure~\ref{fig:decision-tree}.
In general, decision trees are binary trees whose inner nodes are labeled with predicates from an a~priori fixed set $\mathcal P$ and whose leaves are labeled with Boolean values.
In this context, each predicate is a function $p \colon \mathcal D \to \mathbb B$ that maps data points to a Boolean values and corresponds to a property of interest.
Typically, the set $\mathcal P$ is finite, but Ezudheen et al.'s algorithm can build decision trees even from infinite sets of predicates if the underlying domain is numeric.
The algorithm we design in this section retains this feature.
\begin{example}
% !Tex root=main.tex

%---------- Illustration of Decision Tree  ----------

Consider Example \ref{ex:games-and-winning-regions} from Section \ref{sec:framework}.
Remember that in this game a robot moves in an unbounded, one-dimensional grid world. 
Both players can move the robot in an horizontal direction by one cell and are taking turns controlling the robot.
\playera's objective is to keep the robot inside or right of the cell at position $0$.
An illustration of this game is shown in Figure \ref{fig:Example}, where the initial starting position of the robot is shown at $x$.

A possible encoding uses two variables $x\in \mathbb R, y \in \{0,1\}$, as described next.
The variables $x$ corresponds to the $x$-coordinate, while  the variable $y$ indicates which player is currently in control of the robot.
In the case of $y = 0$, \playera{} is in control, thus $V_0= \{(x,y) \in \mathbb{R} \times \{0,1\} \mid y = 0 \}$.
On the other hand, if $z = 1$, then \playerb\ is in control of the robot, thus $V_1 = \{ (x,y) \in \mathbb{R} \times \{0,1\} \mid y = 1 \}$.
Since the starting position is arbitrary within the interval $[0,1)$ which is controlled by \playera , which means that the set of initial vertices is $I = \{(x,y) \in \mathbb{R} \times \{0,1\} \mid x \in [0,1) \land y = 0 \}$.
Furthermore, we fixed the set of safe vertices to be at position greater or equal $0$, thus $F = \{(x,y) \in \mathbb{R} \times \{0,1\} \mid x \geq 0 \}$.
Finally, it is left to define the edge relation.
The robots movement to the right, for instance, can be defined by the relation $R \coloneqq \{ ((x_1,y_1),(x_2,y_2)) \in (\mathbb{R} \times \{0,1\})^2  \mid  y_2 = 1 - y_1, x_2 = x_1+1$.
The direction to the left can be encoded analogously, and the edge relation $E$ is the union of both directions.

Figure \ref{fig:decisiontree-boxgame} depicts the set $\{x \in [0,3)\}$ of Example \ref{ex:learning}.
We first observe that the predicate at the root node is $x < 0$.
Since $x$ determines the position of the robot,  the left subtree (i.e., where $x < 0$ is true) encodes vertices with position $ x < 0 $, while the right subtree (i.e., where $x < 0$ is false) encodes vertices with position $ x\geq 0$.

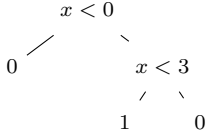
\begin{figure}[h]
\centering
	\begin{tikzpicture}
		\begin{scope}[every node/.append style={font=\strut\footnotesize}]
			\node (0) at (0, 0) {$x < 0$};
			\node (1) at (-1, -.75) {$0$};
			\node (2) at (1, -.75) {$x < 3$};
			\node (3) at (.5, -1.5) {$1$};
			\node (4) at (1.5, -1.5) {$0$};
		\end{scope}

		\draw (0) -- (1) (0) -- (2) (2) -- (3) (2) -- (4);

	\end{tikzpicture}
			\caption{Winning set of Example \ref{ex:games-and-winning-regions} represented as decision tree. }\label{fig:decisiontree-boxgame}	
\end{figure}

Let us first consider the left subtree.
There is only one path leading to the unique leaf node labeled with $0$.
All of these vertices are not safe and correspond to the white area left of position $x = 0$ in Figure~\ref{fig:WinningSet}.
Conversely, only the vertices satisfying $x \geq 0$ and $x < 3$ lead to the unique leaf node labeled with $1$ in the right subtree.
This corresponds to the gray-shaded area in Figure~\ref{fig:WinningSet}.
All of these vertices are safe.
It is not hard to verify that \playera\ can stay within these sets, whereas \playerb\ cannot force a play to an outside vertex.
Moreover, the set of initial vertices is included.
Thus, the learned decision tree in fact encodes a winning set.

Note that in the definition of the winning set requires $W \subseteq V$ to hold.
However, this is not true in this example since the decision tree evaluates to $1$ even in situations where $y \notin \{ 0, 1 \}$, while $V$ contains only vertices with $y \in \{0,1\}$.
Note that this is not a problem in our example since the only checks that require $W \subseteq V$ to hold true are the checks for existential and universal closedness.
Those checks require the winning set $W$ to be intersected with either $V_0$ or $V_1$, which effectively ignores all vertices that violate $y \in \{ 0, 1 \}$.

In conclusion, the winning set depicted in Figure \ref{fig:WinningSet} consists of the 3 cells in the interval $[0,3)$.
The decision tree of Figure~\ref{fig:decisiontree-boxgame} describes this set in an easy to understand manner.

\end{example}
In the remainder of this section, we view a decision tree $t$ as a representation of an (infinite) set $D(t) \subseteq \mathcal D$ of data points.
Whether a data point $d \in \mathcal D$ belongs to this set depends on its \emph{valuation} $t(d) \in \mathbb B$, which is defined as follows: starting at the root node, we recursively descend left (right) if $d$ satisfies (does not satisfy) the predicate at the current node and define $t(d)$ to be the label of the leaf node that is ultimately reached by this procedure.
The set of data points represented by $t$ is then simply the set $D(t) = \{ d \in \mathcal D \mid t(d) = 1 \}$.
Since the learner proposes this set as a hypothesis to the teacher, we call a decision tree $t$ consistent with a game sample $\mathcal S_G$ if $D(t)$ is consistent with $\mathcal S_G$.

Ezudheen et al.'s algorithm has been developed in the context of software verification and expects so-called Horn samples as input.
Formally, a \emph{Horn sample} is a finite set $\mathcal S_H$ containing \emph{Horn constraints} of the form $(d_1 \land \ldots \land d_n) \rightarrow d$ or $(d_1 \land \ldots \land d_n) \rightarrow \mathit{false}$ where $i \in \mathbb N$ and $d, d_1, \ldots, d_n \in \mathcal D$ are data points.
Note that the left-hand-side of a Horn constraint might be empty (i.e., $i=0$), in which case it is interpreted as $\mathit{true}$.

Given a Horn sample $\mathcal S_H$, Ezudheen et al.'s algorithm learns a decision tree $t$ that is consistent with $\mathcal S_H$ in the sense that 
\begin{enumerate*}[label={(\arabic*)}]
	\item $\{ d_1, \ldots, d_n \} \subseteq D(t)$ implies $d \in D(t)$ for each Horn constraint of the form $(d_1 \land \ldots \land d_n) \rightarrow d$ in $\mathcal S_H$ and
	\item $\{ d_1, \ldots, d_n \} \not\subseteq D(t)$ for each Horn constraint of the form $(d_1 \land \ldots \land d_n) \rightarrow \mathit{false}$ in $\mathcal S_H$;
\end{enumerate*}
by extension, we then also say that $D(t)$ is consistent with $S_H$.
Should no consistent decision tree exist (e.g., because the set $\mathcal P$ does not contain sufficient predicates to separate positive and negative examples), the algorithm aborts with an error.
The theorem below summarizes these key properties.

\begin{theorem}[Ezudheen et al.~\cite{DBLP:journals/pacmpl/EzudheenND0M18}] \label{thm:horn-dt}
Let $\mathcal S_H$ be a Horn sample and $\mathcal P$ a finite set of predicates, both over the domain $\mathcal D$.
Moreover, let $n$ be the number of data points and $k$ the number of Horn constraints in $\mathcal S_H$.
If a decision tree over $\mathcal P$ exists that is consistent with $\mathcal S_H$, then Ezudheen et al.'s algorithm learns one in time $\mathcal O(n |\mathcal P| + n^2k)$, assuming that predicates can be evaluated in constant time.
\end{theorem}

\begin{algorithm}[t]
	\DontPrintSemicolon
	\KwIn{A game sample $\mathcal S_G = (\Pos, \Neg, \Ex, \Un)$ and a (finite) set $\mathcal P$ of predicates (both over $\mathcal D$)}
	\BlankLine
	
	{Construct a Horn sample $\mathcal S_H$ as follows:
		\begin{itemize}
		\item for each positive example $d \in \Pos$, \\add the Horn constraint $d \rightarrow \mathit{false}$;
		\item for each negative example $d \in \Neg$, \\add the Horn constraint $\mathit{true} \rightarrow d$;
		\item for each existential implication $d \rightarrow (d_1 \lor \ldots \lor d_n) \in \Ex$, add the Horn \\ constraint $(d_1 \land \ldots \land d_n) \rightarrow d$; and
		\item for each universal implication \\ $d \rightarrow (d_1 \land \ldots \land d_n) \in \Un$, \\ add the Horn constraints \\ $d_1 \rightarrow d, \ldots, d_n \rightarrow d$ to $\mathcal S_H$.
	\end{itemize}}
	
	\BlankLine
	Apply Ezudheen et al.'s learning algorithm to learn a decision tree $t_H$ over $\mathcal P$ that is consistent with $\mathcal S_H$.\;
	
	\BlankLine
	Swap the labels of all leaves of $t_H$ to obtain the decision tree $t_G$ satisfying $D(t_G) = \mathcal D \setminus D(t_H)$ (i.e., the decision trees $t_H$ and $t_G$ are structurally identical but every label $b \in \mathbb B$ in $t_H$ is replaced with $1 - b$).\;
	
	\caption{Learning decision trees from game samples} \label{alg:learner}
\end{algorithm}

We are now ready to present our algorithm for learning decision trees from game samples, which is shown in pseudo code as Algorithm~\ref{alg:learner}.
It builds on top of Ezudheen et al.'s algorithm and proceeds in three steps: first, it translates a game sample $\mathcal S_G$ into an ``equivalent'' Horn sample $\mathcal S_H$; then, it applies Ezudheen et al.'s learning algorithm to obtain a decision tree $t_H$ that is consistent with $\mathcal S_H$; finally, it translates $t_H$ into a decision tree $t_G$, which is consistent with $\mathcal S_G$.
More precisely, the Horn sample $S_H$ constructed in Step~1 has the property that for every decision tree $t_H$ the set $D(t_H)$ is consistent with $S_H$ if and only if its complement $\mathcal D \setminus D(t_H)$ is consistent with the game sample $S_G$.
Since the decision tree $t_G$ obtained in Step~3 of our algorithm satisfies $D(t_G) = \mathcal D \setminus D(t_H)$, it is thus consistent with $\mathcal S_G$.
This property is formalized next.

\begin{lemma} \label{lem:transformation-correct}
Let $\mathcal S_G$ be a game sample and $\mathcal P$ a finite set of predicates, both over the domain $\mathcal D$.
Moreover, let $\mathcal S_H$, $t_H$, and $t_G$ be as in Algorithm~\ref{alg:learner}.
Then, $D(t_H)$ is consistent with $\mathcal S_H$ if and only if $D(t_G)$ is consistent with $\mathcal S_G$.
\end{lemma}
% !Tex root=main.tex
%---------- Proof of Lemma 1 ----------
\label{app:transformation-correct}

\begin{proof}[Proof of Lemma~\ref{lem:transformation-correct}]
Since consistency is a conjunction of conditions for individual examples and Horn constraints, respectively, we show Lemma~\ref{lem:transformation-correct} for each element of a sample individually.
Recall that $D(t_G) = \mathcal D \setminus D(t_H)$.
\begin{itemize}
	\item Let $d \in \Pos$ be a positive example and $d \rightarrow \mathit{false}$ the corresponding Horn constraint generated in Step~1.
	Then,
	\begin{equation*}
	\begin{gathered}
		\text{$t_H$ is consistent with $d \rightarrow \mathit{false}$}  \Leftrightarrow \text{$d \notin D(t_H)$}
		\\ \Leftrightarrow \\ \text{$d \in D(t_G)$}
		 \Leftrightarrow  \text{$t_G$ is consistent with $d \in \Pos$.}
	\end{gathered}
	\end{equation*}
	\item Let $d \in \Neg$ be a negative example and $\mathit{true} \rightarrow d$ the corresponding Horn constraint generated in Step~1.
	Then,
	\begin{equation*}
	\begin{gathered}
		\text{$t_H$ is consistent with $\mathit{true} \rightarrow d$} \Leftrightarrow  \text{$d \in D(t_H)$} 
		\\ \Leftrightarrow \\ \text{$d \notin D(t_G)$} 
		 \Leftrightarrow  \text{$t_G$ is consistent with $d \in \Neg$.}
	\end{gathered}
	\end{equation*}
	\item Let $d \rightarrow (d_1 \lor \ldots \lor d_n) \in \Ex$ be an existential counterexample and $(d_1 \land \ldots \land d_n) \rightarrow d$ the corresponding Horn constraint generated in Step~1.
	Then,
	\begin{equation*}
		\begin{gathered}
		\text{$t_H$ is consistent with $(d_1 \land \ldots \land d_n) \rightarrow d$} \\
		\Leftrightarrow \\
		\text{$\{ d_1, \ldots, d_n \} \subseteq D(t_H)$ implies $d \in D(t_H)$} \\
		\Leftrightarrow \\
		\text{ there exists an $i \in \{ 1, \ldots, n \} $ with} \\  \text{$d_i \notin D(t_H)$ or $d \in D(t_H)$} \\
		\Leftrightarrow \\
		\text{ there exists an $i \in \{ 1, \ldots, n \}$  with} \\ \text{$d_i \in D(t_G)$ or $d \notin D(t_G)$} \\
		\Leftrightarrow \\
		\text{$t_G$ is consistent with $d \rightarrow (d_1 \lor \ldots \lor d_n)$.}
		\end{gathered}
	\end{equation*}
	\item Let $d \rightarrow (d_1 \land \ldots \land d_n) \in \Un$ be a universal implication and $d_1 \rightarrow d, \ldots, d_n \rightarrow d$ the corresponding Horn constraints generated in Step~1.
	Then,
	\begin{equation*}
		\begin{gathered}
		\text{$t_H$ is consistent with $d_1 \rightarrow d, \ldots, d_n \rightarrow d$} \\
		\Leftrightarrow \\
		\text{ $d_i \notin D(t_H)$ for each $i \in \{ 1, \ldots, n \}$ or $d \in D(t_H)$} \\
		\Leftrightarrow \\ 
		\text{ $d_i \in D(t_G)$ for each $i \in \{ 1, \ldots, n \}$ or $d \notin D(t_G)$} \\
		\Leftrightarrow \\
		\text{$t_G$ is consistent with $d \rightarrow (d_1 \land \ldots \land d_n)$.}
		\end{gathered}
	\end{equation*}
\end{itemize}
In total, we obtain that $D(t_G)$ is consistent with $\mathcal S_G$ if and only if $D(t_H)$ is consistent with $\mathcal S_H$.
\end{proof}
The correctness of Algorithm~\ref{alg:learner}, stated in Theorem~\ref{thm:game-dt} below, is now a direct consequence of Lemma~\ref{lem:transformation-correct} and Theorem~\ref{thm:horn-dt}.
Note that the term $n|\Un|$ in the runtime estimation stems from the fact that each universal counterexample might have $n$ data points on its right-hand-side, resulting in $n$ Horn constraints.

\begin{theorem} \label{thm:game-dt}
Let $\mathcal S_G = (\Pos, \Neg, \Ex, \Un)$ be a game sample and $\mathcal P$ a finite set of predicates, both over the domain $\mathcal D$.
Moreover, let $n$ be the number of data points in $\mathcal S_G$.
If a decision tree over $\mathcal P$ exists that is consistent with $\mathcal S_G$, then Algorithm~\ref{alg:learner} learns one in time $\mathcal O\bigl(n |\mathcal P| + n^2(|\Pos| + |\Neg| + |\Ex| + n|\Un|) \bigr)$, assuming that predicates can be evaluated in constant time.
%If no consistent decision tree over $\mathcal P$ exists, the algorithm reports an error.
\end{theorem}

If the abstract domain $\mathcal D$ is numeric, say $\mathcal D \subseteq \mathbb R^n$ for some $n \in \mathbb N$, a simple extension of Ezudheen et al.'s algorithm, and hence Algorithm~\ref{alg:learner}, can be used to learn decision trees even over infinite sets of predicates. 
This extension was originally proposed by Quinlan~\cite{DBLP:books/mk/Quinlan93} and assumes the set of predicates to be $\mathcal P^\star = \{ d[i] \leq c \mid c \in \mathbb R, i \in \{1, \ldots, n\} \}$ where $d[i]$ denotes the $i$-th component of the data point $d$.
The key idea is in fact straightforward: to separate two data points $d_1, d_2 \in \mathcal D$ (e.g., a positive and negative data point), it suffices to consider only such predicates $d[i] \leq c$ for which the constant $c$ occurs as an actual value in $d_1$ or $d_2$.
Thus, Algorithm~\ref{alg:learner} can always restrict itself to a finite subset of $\mathcal P^\star$, which only depends on the values of the data points in the given sample.

In many situations, Algorithm~\ref{alg:learner} in fact guarantees that the overall feedback loop of Section~\ref{sec:framework} converges to a winning set in finite time if one can be expressed as a decision tree.
For instance, if the set $\mathcal P$ is finite, then there exist only finitely many semantically different decision trees (in terms of the set of data points they represent).
Since Algorithm~\ref{alg:learner} always produces consistent decision trees, it will have exhausted all incorrect trees after a finite amount of time.
If a winning set exists and can be expressed as a decision tree over $\mathcal P$, then the subsequent hypothesis will necessarily be one.
On the other hand, if the underlying domain is numeric and Algorithm~\ref{alg:learner} operates over the infinite set $\mathcal P^\star$, a technique proposed by Ezudheen et al.\ can be used to guarantee that a winning set will be learned in finite time (if one can be expressed as a decision tree over $\mathcal P^\star$).
Due to the limited space, we have to refer the reader to Ezudheen et al.~\cite{DBLP:journals/pacmpl/EzudheenND0M18} for details and can here only state our main result.

\begin{theorem}
Let $\mathcal G$ be a safety game.
If the learner of Section~\ref{sec:framework} uses Algorithm~\ref{alg:learner} over a finite set $\mathcal P$ of predicates, then it is guaranteed to learn a winning set after a finite number of iterations if there exists one that is expressible as a decision tree over $\mathcal P$.
The analogous statement holds for the set $\mathcal P^\star$.
\end{theorem}

In practice, the overall runtime depends not only on the learner but also on the teacher (which might be adversarial). 
Hence, an average case analysis requires further assumptions on the teacher.
We leave such an analysis for future work and turn to an experimental evaluation instead.

%---------- Experimental Evaluation ----------
% !Tex root=main.tex

\tikzset{
  head/.style = {fill = orange!90!blue,
                 label = center:\textsf{\Large H}},
  tail/.style = {fill = blue!70!yellow, text = black,
                 label = center:\textsf{\Large T}}
}

\section{Experimental Evaluation}
\label{sec:experiments}

To assess the performance of our learning framework and the decision tree learning algorithm, we have implemented a prototype named \dtsynth.\footnote{Code/benchmarks available at \url{https://github.com/OliverMa1/DT-Synth}.}
In this section, we describe \dtsynth in detail and compare it to three existing tools: \consynth~\cite{DBLP:conf/popl/BeyeneCPR14}, which is based on constraint solving, as well as two tools based on learning finite automata~\cite{DBLP:conf/tacas/NeiderT16}, for brevity here called \satsynth and \rpnisynth.

%---------- The Tools ----------
\subsection{Tools}
\dtsynth takes games as input that are encoded as quantifier-free, first-order formulas in the theory of linear integer arithmetic (LIA).
More precisely, variables $x_1, \ldots, x_n$ of type integer encode vertices of a game graph, while formulas $\varphi(x_1, \ldots, x_n)$ over those variables represent (infinite) sets of vertices.
Additionally, the edge relation of a game graph is encoded by a formula $\psi(x_1, \ldots, x_n, x'_1, \ldots, x'_n)$ where $x_1, \ldots, x_n$ represents a source-vertex and $x'_1, \ldots, x'_n$ a destination-vertex.
We believe that this encoding of games is highly accessible and allows modeling many real-word synthesis tasks (e.g., robotic motion planning) in a natural way.
Moreover, the decision trees computed by \dtsynth are generally easy for humans to comprehend, as we demonstrate in Section \ref{sec:dt-learner}.
Note, however, that the choice of LIA restricts our game graphs to countably many vertices.
We have made this choice deliberately to be able to reuse  Ezudheen et al.'s implementation of the decision tree learner~\cite{DBLP:journals/pacmpl/EzudheenND0M18}. 
\dtsynth implements the teacher and learner as follows:
\paragraph*{Teacher}
	The teacher builds on top of the Z3 SMT solver~\cite{DBLP:conf/tacas/MouraB08}.
	Upon receiving a hypothesis in the form of a formula $\varphi_H(x_1, \ldots, x_n)$, it performs a series of satisfiability checks according to Definition~\ref{def:teacher} in order to search for counterexamples.
	If a satisfiability check succeeds, the teacher derives a counterexample from the model returned by Z3 (potentially triggering a finite number of additional satisfiability checks to compute the successors of a vertex in the case of existential and universal counterexamples). 
\paragraph*{Learner}
	The learner builds on top of code originally developed by Ezudheen et al.~\cite{DBLP:journals/pacmpl/EzudheenND0M18}.
	In addition to the set $\mathcal P^\star$ of predicates described in Section~\ref{sec:dt-learner}, it additionally uses octagonal constraints of the form $x_i \pm x_j \leq c$ (which capture distances between vertices in the Euclidean space).
	After a consistent decision tree has been learned, the learner converts it into a formula $\varphi_H \coloneqq \bigvee_{\pi \in \Pi}\bigwedge_{\psi \in \pi} \psi$ where $\Pi$ is the set of all paths from the root to a leaf labeled with $\mathit{true}$, and $\psi \in \pi$ denotes that the predicate $\psi$ occurs on $\pi$ (negated if the path descends to the right).
	Then, it hands $\varphi_H$ over to the teacher.
	Following the description in Section~\ref{sec:dt-learner}, it is not hard to verify that a vertex satisfies the formula $\varphi_H$ obtained from the tree $t$ if and only if the vertex belongs to $D(t)$. \par\vskip\baselineskip

We compared \dtsynth to three other tools: \consynth, \satsynth and \rpnisynth.
Table~\ref{tab:tool-comparison} summarizes the main similarities and differences of these tools.

\begin{table}
	\caption{Properties of the compared tools} \label{tab:tool-comparison}
	\centering
	\begin{tabular}{l@{\hskip 2em}cccc}
		\toprule
		\multirow{3}{*}{Tool}  & Easy to & Easy to & Guarantees & No help\\
		& model & interpret & to find a & of user \\
		& games & solution & strategy & required\\
		\midrule
		\dtsynth & \ding{51} & \ding{51} & \ding{51} & \ding{51} \\
		\consynth~\cite{DBLP:conf/popl/BeyeneCPR14} & \ding{51} & \ding{55} &\ding{51}& \ding{55} \\
		\satsynth~\cite{DBLP:conf/tacas/NeiderT16} & \ding{55} & \ding{55} & \ding{51} & \ding{51} \\
		\rpnisynth~\cite{DBLP:conf/tacas/NeiderT16} & \ding{55} & \ding{55} & \ding{55} & \ding{51} \\
		\bottomrule
	\end{tabular}
\end{table}

\consynth~\cite{DBLP:conf/popl/BeyeneCPR14} is a tool for synthesizing winning strategies in games over infinite graphs with $\omega$-regular winning conditions (which are more general than the safety games considered in this paper).
The tool reduces the computation of a winning strategy into a satisfiablity problem for Constrained Horn Clauses (CHCs)~\cite{DBLP:conf/cav/BeyenePR13}, which allows it to leverage the power of modern SMT solvers.
Similar to \dtsynth, \consynth's input is a game (with potentially unaccountably many vertices) encoded as an SMT formula.
However, to remove existential quantifiers that arise from translating the synthesis problem into CHCs, \consynth relies on so-called \emph{Skolem templates}, which have to be specified by the user (we discuss the impact of Skolem templates later in this section).
Intuitively, these templates constrain the search space of potential strategies and, hence, require insight into what a winning strategy might be (and the ability to express this insight as a Skolem template).
By contrast, \dtsynth does not require additional help from the user (it synthesizes winning sets\slash strategies based on the game alone), but it can also not benefit from the user's knowledge.
The output of \consynth is values for the parameters in the Skolem template (encoding a winning strategy) and a corresponding winning set in form of an SMT formula.
Both objects can be challenging for humans to understand.
Similar to \dtsynth, \consynth guarantees to find a winning strategy if one can be expressed in terms of the Skolem template.

\satsynth and \rpnisynth~\cite{DBLP:conf/tacas/NeiderT16} are two techniques for synthesizing winning strategies of safety games over game graphs with countably many vertices.
Both tools expect the game encoded by means of finite automata/transducers and leverage automata learning techniques to compute winning sets: \satsynth internally uses a SAT solver to learn automata, wheres \rpnisynth is based on the popular RPNI learning algorithm~\cite{RPNI}.
Similar to \dtsynth and unlike \consynth, both tools do not require additional information from the user.
However, both the input and output of \satsynth and \rpnisynth are finite automata, which is not a very natural encoding of games and can be difficult for humans to understand.
Finally, \satsynth guarantees to find a winning set (if it can be expressed as an automaton), whereas \rpnisynth does not.

\begin{table*}[t]
	\caption{
		Experimental results on the \satsynth/\rpnisynth benchmark suite~\cite{DBLP:conf/tacas/NeiderT16} (upper part) and \consynth benchmark suite~\cite{DBLP:conf/popl/BeyeneCPR14} (lower part).
		``Iter.'' refers to the number of iterations in the counterexample-guided feedback loop.
		``Size'' measures the size of the final decision tree learned by \dtsynth in terms of inner nodes and the size of the final automata produced by \satsynth and \rpnisynth in terms of the number of states, respectively.
		``---'' indicates a timeout after $900\,s$.
	}
	\label{tab:experiment1}

	\centering
\begin{tabular}{lrrr@{\hskip 2em}rrr@{\hskip 2em}rrr@{\hskip 2em}r}
\toprule
	& \multicolumn{3}{c}{\hskip -1em \dtsynth} & \multicolumn{3}{c}{\hskip -1em \satsynth} & \multicolumn{3}{c}{\hskip -1em \rpnisynth} & \consynth \\ \cmidrule(lr{2em}){2-4} \cmidrule(l{-0.25em}r{2em}){5-7} \cmidrule(l{-0.25em}r{2em}){8-10} \cmidrule(lr){11-11}
	Game  & Time in s& Iter.& Size &Time in s &	Iter.& Size & Time in s& 		Iter.& Size  & Time in s\\
	\midrule
Box &$ 0.77$ & $9$ & $5$ &$1.22$ & $45$ & $5$ &$0.41$ & $16$ & $6$ & $3.71$\\
Box Limited & $0.28$ & $4$ & $2 $ &$0.74$ & $36$ & $4  $&$0.32$ & $15$ & $5$ & $0.44$ \\
Diagonal & $2.04$ & $23 $& $5$  & $1.49$ & $64$ & $4$ & $1.01$ & $64$ & $4$  & $1.93$\\ 
Evasion & $0.63$ & $6$ &$ 3$  &$89.15$ & $237$ & $7$ & $1.40$ & $82$ & $11$ & $1.50$\\ 
Follow & $0.86$ & $11$ & $5$ & $95.90$ & $300$ & $7$  &$ 7.36$ & $352$ & $16$ &---\\ 
Solitary Box & $0.24$ & $4$ & $2$ &$5.81$ & $76$ & $6$  &$0.42$ & $16$ & $6$  & $0.42$\\
Square $5x5$ & $7.77$ & $61$ & $12$  &\multicolumn{3}{c}{---------------------------}& $0.82$ & $39$ & $14$ & --- \\
\midrule
Cinderella ($c=2$) & \multicolumn{3}{c}{---------------------------}  & \multicolumn{3}{c}{---------------------------} & \multicolumn{3}{c}{---------------------------}&---\\
Cinderella ($c=3$) & \multicolumn{3}{c}{---------------------------}  & \multicolumn{3}{c}{---------------------------} & \multicolumn{3}{c}{---------------------------} &$765.30$\\
Program-repair & $0.99$ & $14$ & $11$  &$1.42$ & $69$ & $3$  & $0.15$ & $7$ & $3$ & $2.50$ \\
Repair-critical & $23.74$ & $237$ & $14$   & \multicolumn{3}{c}{---------------------------} & $130.54$ & $1772$ & $11$ &$19.52$ \\
Synth-Synchronization & $63.30$ & $513$ & $35$  & \multicolumn{3}{c}{---------------------------} & $42.74$ & $858$ & $26$ &$10.01$ \\
	\bottomrule
\end{tabular}
\end{table*}
%---------- The Benchmarks ----------
\subsection{Benchmarks}
We have evaluated the performance of all four tools on two benchmark suits, both featuring safety games over infinite game graphs.%
\footnote{Note that our approach is designed for infinite graphs but not to compete with highly-optimized synthesis engines on games over finite graphs. Hence, we did not consider benchmarks from the various synthesis competitions.}
The first benchmark suite accompanies \satsynth/\rpnisynth~\cite{DBLP:conf/tacas/NeiderT16} and consists of seven safety games, which are motivated by  robotic motion planning.
The second suite of benchmarks is shipped with \consynth~\cite{DBLP:conf/popl/BeyeneCPR14} and consists of five safety games: two versions of a combinatorial puzzle (Cinderella game), two program repair problems, and a synchronization problem for multi-threaded programs.
% !Tex root=main.tex

%---------- Description of benchmarks  ----------
\label{app:benchmark-description}

\paragraph*{First Benchmark Suite}
Our first benchmark suit comprises the following seven games, taken from Neider and Topcu~\cite{DBLP:conf/tacas/NeiderT16}.
Most games are from the area of robotic motion planning, and all are over infinite graphs.
\begin{description}[font={\normalfont\itshape}]
	\item[Diagonal game:] A robot moves in an infinite, discrete two-dimensional grid world. \playera\ controls the robot's vertical movement, while \playerb\ controls the horizontal. \playera\ wins if the robot stays within two cell around the diagonal.
	\item[Box game:] A variation of the diagonal game. Both players can move the robot in an vertical, horizontal or diagonal direction by one cell. \playera\ wins if the robot stays within a horizontal stripe of width three.
	\item[Limited Box game:] A variation of the box game. \playera\ can only control the robot's vertical movement and \playerb\ the horizontal.
	\item[Solitary box game:] Another variation of the Box game in which only \playera\ is in control of the robot.
	\item[Evasion game:] Two robots are moving in an infinite, discrete two-dimensional grid world. The robots take turns moving at most one cell in any direction. Each players controls one robot. \playera's objective is to avoid getting caught by \playerb's robot.
	\item[Follow game:] A version of the evasion game where \playera's goal is to keep its robot within a Manhattan distance of two cells to the environment's robot.
	\item[Square game:] A variation of the box game, where \playera\ wins if the robot stays within a fixed size square (here $5 \times 5$).
\end{description}

\paragraph*{Second Benchmark Suite}
Our second benchmark suit comprises the following five games, taken from  Beyene et al.~\cite{DBLP:conf/popl/BeyeneCPR14}.
Most games are from the area program synthesis and program repair, thus involve one player only.
\begin{description}[font={\normalfont\itshape}]
	\item[Cinderella game:] 
	The Cinderella game is a turn-based game, originally posed as a challenge for the synthesis community~\cite{DBLP:conf/ifipTCS/BodlaenderHKSWZ12}.
	It involves the protagonist, the mythical Cinderella, and the antagonist, her stepmother.
	The game starts with five empty buckets, which are arranged in a circle and can hold up to a constant $c$ units of water.
	In every round of the game, the stepmother brings one unit of water, which she distributes arbitrarily among the five buckets.
	Then, Cinderella is allowed to empty two adjacent buckets.
	The stepmother wins if she can make one of the buckets overflow.
	Cinderella, on the other hand, wins if she can indefinitely prevent all buckets from overflowing.
	Beyene et al.\ considered two versions of the game for $c=2$ and $c=3$, respectively, and have discretized the game such that the stepmother can distribute the water only in $0.1$ units.
	\item[Program repair game:] The program-repair game looks for a modification of statements such that the modified program satisfies its specification.
	\item[Repair-critical game:] The Repair-critical game is a game derived from concurrent program repair problems under fairness assumption.
	\item[Synthesis Synchronization game:] This game is an example for synthesis of synchronization in multi-threaded programs.
\end{description}

For the first benchmark suite, we have equipped \consynth with moderately restrictive Skolem templates; templates for the second suite were provided by Beyene et al.
The representations of winning sets\slash strategies for all four tools are expressive enough for all games in these benchmark suits.

%---------- Results ----------
\subsection{Results}
Table~\ref{tab:experiment1} lists the experimental results for all four tools on the \satsynth/\rpnisynth benchmark suite~\cite{DBLP:conf/tacas/NeiderT16} (upper part) and the \consynth benchmark suite~\cite{DBLP:conf/popl/BeyeneCPR14} (lower part).%
\footnote{Note that we were not able to reproduce the exact results for \consynth reported by Beyene et al.~\cite{DBLP:conf/popl/BeyeneCPR14}. \consynth has numerous options, which influence the performance of the benchmarks drastically. We have been in contact with the authors to obtain the exact options that were used in their experiments. Unfortunately, some of these options resulted in crashes and incorrect results on our machine. Hence, we resorted to the default options.}
In the case of \dtsynth, \rpnisynth, and \satsynth, Table~\ref{tab:experiment1} also shows the number of iterations as well as the size of the result (measured in the number of inner nodes in a decision tree and the number of states of an automaton, respectively).
We have conducted all experiments on an Intel Xeon E7-8857 v2 CPU with $4$\,GB of RAM running a 64-bit Debian operating system.
The timeout was $900\,s$.

We rank the performance of the tools based on the number of games they can solve.
To break ties, we consider the aggregate runtime on games that the tools were able to solve (i.e., not accounting for time-outs).
With this scoring scheme, \dtsynth performed best and solved ten out of twelve games with an aggregated runtime of $100.62\,s$.
\rpnisynth ranked second, having solved the same ten games but with an aggregated runtime of $185.17\,s$ (i.e., $1.8$ times slower than \dtsynth).
\consynth ranked third and solved nine games with an aggregated runtime of $805.33\,s$.
Finally, \satsynth performed worst, having solved seven games with an aggregate runtime of $195.73\,s$.

Compared to \rpnisynth and \satsynth, \dtsynth required far fewer iterations and, hence, fewer interactions with the typically computationally expensive teacher.
Moreover, the size of the final output is smaller, which makes it easier for humans to interpret.
It is also important to emphasize that despite \rpnisynth's good performance, the tool does not guarantee to find a winning set.
By contrast, \dtsynth provides such a guarantee.

Note that \consynth was the only tool able to solve at least one version of the Cinderella game.
We believe that this is due to the somewhat restrictive Skolem template that Beyene et al.\ have provided: if equipped with less restrictive templates, \consynth times out as well.

In conclusion, \dtsynth is competitive to the state-of-the-art tools for solving safety games over infinite graphs.
It does not require any user guidance, guarantees to find a winning set (if one can be expressed as a decision tree), and features easy-to-understand input/output formats.

% !Tex root=main.tex

%---------- Description of Cinderella game  ----------
\subsection{Impact of skolem templates}
\label{app:cinderella-game}
To assess the impact of Skolem templates on the performance of \consynth, we have conducted a case study based on the Cinderella game.
This case study consists of a series of Skolem templates for Cinderella's strategy that successively permit more and more complex behavior.
Our goal is to determine the point at which Skolem templates become too permissive and \consynth is no longer able solve the game within a reasonable time frame.

In the following, we explain the Skolem templates provided by Beyene et al.~\cite{DBLP:conf/popl/BeyeneCPR14}, describe our less restrictive templates in detail, and finally discuss the outcome of our experiments with these templates.
For the remainder of this subsection, we focus on the game with $c = 3$ as it has a winning strategy for Cinderella that is easy to understand.

%---------- Skolem templates and the Cinderella game ----------
\paragraph{Skolem templates and the Cinderella game}
%Skolem templates have a stark impact on \consynth's ability to find winning strategies and the time needed to compute one.
Beyene et al.'s formulation of (safety) games in terms of Constrained Horn Clauses (CHCs) relies on existential quantifiers to express the effects of actions of \playera.
In order to eliminate these quantifiers and make the resulting formulas amenable to constraint solving, \consynth uses so-called Skolem templates.
Intuitively, Skolem templates are user-provided formulas that capture high-level intuitions about potential winning strategies and, hence, restrict the space of strategies that \consynth has to consider.

In the case of the Cinderella game with $c=3$ the Skolem template provided by Beyene et al.\ only allows quite restrictive strategies: the choice of buckets that Cinderella has emptied in the current round always dictates the choice of buckets that she will empty in the next round, completely ignoring which buckets the stepmother actually fills.
In fact, it turns out that continuously emptying Buckets~$1$ and $2$, then Buckets~$2$ and $3$, and then Buckets~$4$ and $5$ is a winning strategy.

Intuitively, Beyene et al.'s Skolem template can be seen as an encoding of a finite-state machine with five states $\{ s_1, \ldots, s_5 \}$, where the state $s_i$ prescribes that Cinderella empties Buckets $i$ and ${i+1} \mod 5$ (e.g., the winning strategy above corresponds to the finite-state machine depicted in Figure~\ref{subfig:skolem-templates-original}).
In the Skolem template, the transitions of the finite-state machine are left undefined, and solving the game amounts to directing the transitions in a way that the resulting strategy is winning for Cinderella.
%in the strategy above, the initial state of the machine is $s_1$ and the transitions are $s_1 \to s_2$, $s_2 \to s_4$, and $s_4 \to s_1$.
However, without solving the game first, it is very hard to know beforehand that a winning strategy of this form actually exists. 

\begin{figure}
	\centering
	\begin{subfigure}{.485\textwidth}
		\centering
		\begin{tikzpicture}
			\begin{scope}[every state/.append style={font=\small}]
				\node[state] (1) at (0, 0)  {$s_1$};
				\node[state] (2) at (1.5, 0)  {$s_2$};
				\node[state] (3) at (3, 0)  {$s_3$};
				\node[state] (4) at (4.5, 0)  {$s_4$};
				\node[state] (5) at (6, 0)  {$s_5$};
				\draw[<-, shorten >=0pt, shorten <=1pt] (1.west) -- +(-.3, 0);
				\path[->] (1) edge (2);
				\path[->] (2) edge[bend left=45] (4);
				\path[->] (4) edge[bend left=45] (1);
			\end{scope}
		\end{tikzpicture}
		\subcaption{Beyene et al.'s original Skolem template with transitions corresponding to a winning strategy} \label{subfig:skolem-templates-original}
	\end{subfigure}
	\begin{subfigure}{.485\textwidth}
		\centering
		\begin{tikzpicture}
			\begin{scope}[every state/.append style={font=\small}]
				\node[state] (1) at (0, 0)  {$s_1$};
				\node[state] (1a) at (1.5, 0)  {$s'_1$};
				\node[state] (2) at (3, 0)  {$s_2$};
				\node[state] (3) at (4.5, 0)  {$s_3$};
				\node[state] (4) at (6, 0)  {$s_4$};
				\node[state] (5) at (7.5, 0)  {$s_5$};
				\draw[<-, shorten >=0pt, shorten <=1pt] (1.west) -- +(-.3, 0);
				\path[->] (1) edge (1a);
				\path[->] (1a) edge (2);
				\path[->] (2) edge[bend left=45] (4);
				\path[->] (4) edge[bend left=45] (1);
			\end{scope}
		\end{tikzpicture}
		\subcaption{An extended Skolem template} \label{subfig:skolem-templates-extended}
	\end{subfigure}
	\caption{Visualization of Skolem templates for the Cinderella game}
\end{figure}
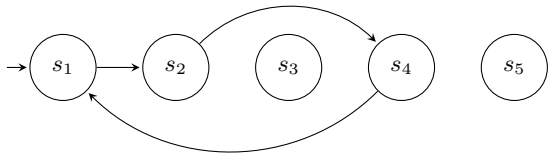
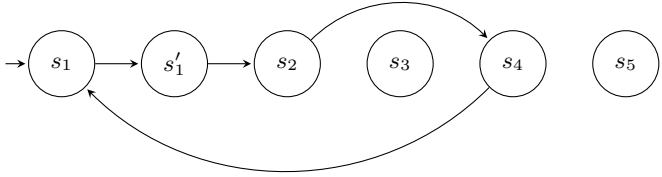

%---------- Constructing less restrictive templates ----------
\paragraph{Constructing less restrictive templates} 
Generating less restrictive Skolem templates is now straightforward:
instead of five states, we successively increase the number of states, thus allowing for incrementally more complex behavior of Cinderella.
For instance, the finite-state machine in Figure~\ref{subfig:skolem-templates-extended} prescribes Cinderella to empty Buckets $1$ and $2$ twice before proceeding to emptying Buckets $2$ and $3$ as well as Buckets $4$ and $5$.
Note, however, that the strategy obtained from Beyene et al.'s original Skolem template always remains possible (as long as the encoded finite-state machine contains at least the states $s_1$, $s_2$, and $s_4$).

%---------- Impact of less restrictive Skolem templates ----------
\paragraph{Impact of less restrictive Skolem templates}
We have run \consynth on the Cinderella game with $c=3$ for a series of Skolem templates that encode finite-state machines of increasing size.
Recall that \consynth computed a winning strategy using Beyene et al.'s original template (consisting of five states) in $765.30\,s$.
For a template encoding six states (one more state in which Cinderella can empty Buckets $1$ and $2$), \consynth computed a winning strategy in $1209.59\,s$
However, already for a template encoding seven states (the six-state template above plus an additional state in which Cinderella can empty Buckets $3$ and $4$), \consynth was unable to find a winning strategy within $48\,h$.
In particular, note that giving Cinderella the option to empty Buckets $3$ and $4$ twice increases the search space for strategies unnecessarily because Cinderella never requires  to empty those two adjacent buckets in order to win.

The Skolem templates used in \consynth are  powerful tools to facilitate the computation of winning strategies, even for difficult safety games.
However, designing such templates usually requires (high-level) knowledge about the form of winning strategies and the ability to express this knowledge in terms of logical formulas.
Moreover, our experiments show that using even slightly suboptimal templates (e.g., by allowing one or two too additional moves for \playera) can lead to an increase in computational time that is prohibitive in practice.

%---------- Conclusion ----------
% !Tex root=main.tex

\section{Conclusion}

We have developed a machine learning framework for synthesizing reactive safety controllers whose interaction with their environment is modeled by games over infinite graphs.
Moreover, we have designed a learning algorithm for decision trees that learns winning sets\slash strategies and is in many situations guaranteed to find a solution if one exists.
Our experimental evaluation shows that our approach is highly competitive and promises applicability to a wide range of interesting practical problems, specifically due to its ease of use.

A promising direction for future work would be to apply our technique to distributed synthesis problems and other, more complex synthesis settings.
Moreover, we plan to extend our learning-based framework to more general winning conditions, such as reachability and liveness.

%---------- Acknowledgment ----------
\section*{Acknowledgment}
This work was partially funded by the ERC Starting Grant AV-SMP (grant agreement no.\ 759969)

%---------- Bibliography ----------
\bibliographystyle{IEEEtran}
\bibliography{bib}

\end{document}